\documentclass[5p]{elsarticle}  
\usepackage{amsmath} 
\usepackage{amsthm}

\usepackage{latexsym}
\usepackage{graphicx}

\newtheorem{df}{Definition}
\newtheorem{lem}[df]{Lemma}
\newtheorem{thm}[df]{Theorem}

\begin{document}
\begin{frontmatter}
\title{Oriented chromatic number of Halin graphs}

\author{Janusz Dybizba\'nski}
\ead{jdybiz@inf.ug.edu.pl}
\author{Andrzej Szepietowski}
\ead{matszp@inf.ug.edu.pl}
\address{Institute of Informatics, University of Gda\'nsk\\ Wita Stwosza 57,
80-952 Gda\'nsk, Poland}

\begin{abstract}
Oriented chromatic number of an oriented graph $G$ is the minimum order of an oriented graph $H$ such that $G$ admits a homomorphism to $H$. 
The oriented chromatic number of an unoriented graph $G$ is the maximal
chromatic number over all possible orientations of $G$.
In this paper, we prove that every Halin graph has oriented chromatic number at most 8, improving a previous bound by Hosseini Dolama and Sopena, and confirming the conjecture given by Vignal.
\end{abstract}

\begin{keyword}
Graph coloring \sep oriented graph coloring \sep Halin graph \sep oriented chromatic number 
\end{keyword}
\end{frontmatter}
\section{Introduction}
{\itshape Oriented coloring} is a coloring of the vertices of an oriented graph
$G$ such that: (1) no two neighbors have the same color, (2)
for every two arcs $(t,u)$ and $(v,w)$, either $\beta(t)\ne \beta(w)$ or
$\beta(u)\ne \beta(v)$.
In other words, if there is an arc leading  from the color $\beta_1$ to
$\beta_2$, then no arc leads from $\beta_2$ to $\beta_1$.

It is easy to see that an oriented graph $G$ can be colored by $k$ colors if
and only if there exists a homomorphism from $G$ to an oriented
graph $H$ with $k$ vertices. In this case we shall say that $G$ is colored by $H$.

The {\itshape oriented chromatic number $\overrightarrow\chi(G)$ of an oriented graph $G$}  is the
smallest number $k$ of colors needed to color $G$, and the
 {\itshape  oriented chromatic number $\overrightarrow\chi(G)$ of an unoriented graph $G$}  is the
maximal chromatic number over all possible orientations of $G$.
The {\itshape oriented chromatic number of a family of graphs} is the
maximal chromatic number over all possible graphs of the family .

Oriented coloring has been studied  in recent years
\cite{bor:ii:dis,dol:ii:dis,fer:i:ipl,hos:ii:ipl,kos:ii:gra,ras:ii:ipl,so2:ii:dis,so1:ii:gra,
so4:ii:ipl,so3:ii:bor,sze:ii:ipl},
see \cite{Sop:i:pre} for a short survey of the main results. Several authors
established or gave bounds on the oriented 
chromatic number for some families of graphs, such as: oriented planar graphs
\cite{ras:ii:ipl},
outerplanar graphs 
\cite{so1:ii:gra,so4:ii:ipl},
graphs with bounded degree three 
\cite{kos:ii:gra,so1:ii:gra,so3:ii:bor}, 
$k$-trees \cite{so1:ii:gra}, graphs with given
excess \cite{dol:ii:dis}, grids 
\cite{dyn:i:ipl,fer:i:ipl,sze:ii:ipl} or hexagonal grids \cite{biel:i:ann}.

In this paper we focus on the oriented chromatic number of Halin graphs. A
{\itshape Halin graph} $H$ is an unoriented planar graph which admits a planar embedding such that deleting the edges of its external face ($F_0$) gives a tree with at least three leaves. The vertices on
$F_0$ are called exterior vertices of $H$, and the remaining vertices are
called interior vertices of $H$.

In~\cite{Vig} Vignal proved that every oriented Halin graph has oriented chromatic number at most 11 and conjectured that the oriented chromatic number of every oriented Halin graph is at most $8$. Hosseini Dolama and Sopena proved in \cite{hos:ii:ipl} that every oriented Halin graph has oriented chromatic number at most 9 and they presented an oriented Halin graph with oriented chromatic number equal to 8. Figure~\ref{dol} presents another example of Halin graph with oriented chromatic number equal to 8. Determining the exact value of oriented chromatic number of Halin graph is an open problem presented by Sopena in~\cite{Sop:i:pre}.

\begin{figure}[htb]
\centering
\includegraphics{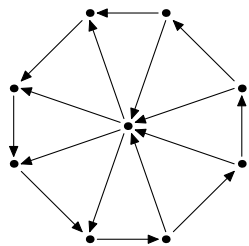}
\caption{Halin graph with oriented chromatic number equal to 8.}
\label{dol}
\end{figure}

In this paper we shall prove that every oriented Halin graph can be colored with at most 8 colors. Hence, the oriented chromatic number of the family of Halin graphs is equal to 8.

\section{Preliminaries} 

We recall that $T_7$ is the tournament build from the non-zero quadratic residues of 7,
see \cite{bor:ii:dis, fer:i:ipl, Fr70, so1:ii:gra, so3:ii:bor}.
More precisely, $T_7$ is the graph with vertex set $\{0, 1, \dots, 6\}$ and such that $(i,j)$
is an arc if and only if $j-i = 1,2,\textrm{ or }4\pmod 7$.

\begin{lem}[see~\cite{sze:ii:ipl}]\label{l3}
For every $a\in\{1,2,4\}$ and $b\in\{0, 1, \dots, 6\}$,
the function $\phi(x)=ax+b \pmod7$ is an automorphism in $T_7$.
\end{lem}

\begin{lem}[see~\cite{sze:ii:ipl}]\label{l1}
Let $G$ be an oriented graph, $h$ a homomorphism from $G$
into $T_7$, and let $G'$ be the oriented graph obtained from $G$ by reversing
all arcs. More precisely, $(u,v)$ is an arc in $G'$ if and only if
$(v,u)$ is an arc in $G$.

Then the function $f(x)=-h(x)\bmod7$ is a homomorphism from $G'$ into $T_7$.
\end{lem}

In the sequel we shall simply write $ax+b$ instead of $ax+b \pmod7$ when writing about homomorphisms of $T_7$.

A fan $F$ is an oriented planar graph which consists of a rooted oriented 
tree with a root $r$ and leaves $x_1,\dots,x_m$; and for every 
$1\le i\le m-1$, the leaves $x_i$ and $x_{i+1}$ are 
connected by an arc: $(x_i,x_{i+1})$ or $(x_{i+1},x_i)$.
We shall denote the root of $F$ by $r(F)$, the first leaf $x_1$ by
$fl(F)$, and the last leaf $x_m$ by $ll(F)$.

Note that if we  remove one vertex or arc from 
the exterior cycle of a Halin graph, then we obtain a fan.
But we shall also consider other fans, e.g  with one leaf, $m=1$, or 
with the root having only one son.

Suppose we have two fans $F_1$ and $F_2$. We can compose them
in one fan $F$, denoted by $F_1+F_2$, in the following way, see Fig. ~\ref{fig1}:

\begin{figure}[tbp]
\centering
\includegraphics{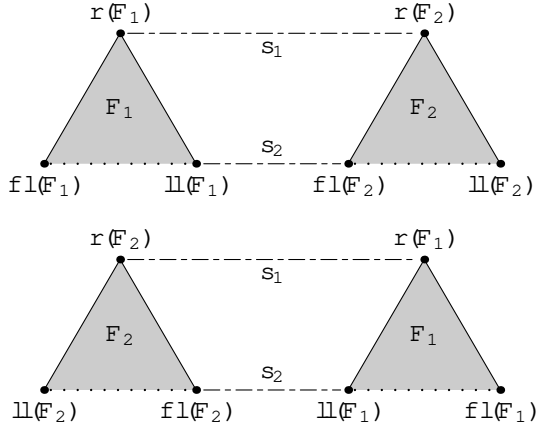}
\caption{Two ways to compose $F_1+F_2$.}
\label{fig1}
\end{figure}

\begin{itemize}
\item root of $F_1$ becomes the root of $F$, i.e. $r(F):=r(F_1)$,
\item $r(F_1)$ is joined with $r(F_2)$ by an arc $s_1$, i.e.
$s_1=(r(F_1),r(F_2))$ or $s_1=(r(F_2),r(F_1))$,
\item $ll(F_1)$ is joined with $fl(F_2)$ by an arc $s_2$,
\item the first leaf of $F_1$ becomes the first leaf of $F$, 
i.e. $fl(F):=fl(F_1)$,
\item the last leaf of $F_2$ becomes the last leaf of $F$, 
i.e. $ll(F):=ll(F_2)$, and $r(F_2)$ becomes the last son of $r(F_1)$.
\end{itemize}

\begin{lem}\label{comp}
Suppose there are two fans $F_1$ and $F_2$ and colorings $c_i:F_i\to T_7$, 
for $i=1,2$, such that
in each fan $F_i$ the root is colored with zero, 
$c_i(r(F_i))=0$, and both the first and the last leaves with non zero,
$c_i(fl(F_i))\ne 0$, $c_i(ll(F_i))\ne0$. 
Then for every direction of the arcs $s_1$ and $s_2$,
the composition $F=F_1+F_2$ can be colored with two colorings
$d_1,d_2: F\to T_7$ such that:
\begin{itemize}
\item[(c1)] $d_1(x)=d_2(x)=c_1(x)$, for all $x\in F_1$
\item[(c2)] $d_1(ll(F))\ne d_2(ll(F))$ 
\end{itemize}
In other words colorings $d_1$ and $d_2$ are equal to $c_1$ on $F_1$
(in particular $d_1(r(F))=d_2(r(F))=0$ and $d_1(fl(F))=d_2(fl(F))\ne0$ ) 
and differ on the last leaf of $F$ (it is possible that one of 
$d_1(ll(F))$ or $d_2(ll(F))$ is equal to zero, but not both).
\end{lem}
\begin{proof}
The main idea of the proof is to find two authomorphisms $\phi_1$ 
and $\phi_2$ of $T_7$ which change the coloring $c_2$ on $F_2$ in 
such a way that the new colorings $\phi_1\circ c_2$ and $\phi_2\circ c_2$
will fit to coloring $c_1$ on $F_1$, and be different from each other 
on $ll(F_2)$. 
More precisely, we shall show that for any colors 
$c_1(ll(F_1))$, $c_2(fl(F_2))$, $c_2(ll(F_2))$ and any direction of 
$s_1$ and $s_2$, there exist two authomorphisms 
$$\phi_1, \phi_2:T_7\to T_7$$
such that the colorings $d_1$ and $d_2$ defined by:

$$d_1(x) =\left\{
\begin{array}{ll}
 c_1(x) &\textrm{ if $x\in F_1$}\\
 \phi_1(c_2(x)) &\textrm{ if $x\in F_2$}
\end{array}
\right.
$$
and

$$d_2(x) =\left\{
\begin{array}{ll}
 c_1(x) &\textrm{ if $x\in F_1$}\\
 \phi_2(c_2(x)) &\textrm{ if $x\in F_2$}
\end{array}
\right.
$$
color the composition $F=F_1+F_2$ in a proper way and satisfy (c2).

We can assume that $s_1$ goes from $r(F_1)$ to $r(F_2)$.
Otherwise we can reverse all arcs, negate all colors, color $F_1+F_2$,
and reverse arcs and negate colors back.
We can also assume that $c_1(ll(F_1))$, and $c_2(fl(F_2))$, are both in 
$\{1,3\}$. Otherwise we can multiply all colors in $F_1$, or $F_2$, by 2 
or 4.

Consider first the case when
$c_1(ll(F_1))=3$, $c_2(fl(F_2))$ $=1$, and $s_2$ goes 
from $ll(F_1)$ to $fl(F_2)$.
In this case we first consider authomorhisms 
$\phi_1(x)=x+4$ and  $\phi_2(x)=2x+2$.
The arc $s_1=(r(F_1),r(F_2))$ has colors $(0,\phi_1(0))=(0,4)$ or 
$(0,\phi_2(0))=(0,2)$ which are proper.
The arc $s_2=(ll(F_1),fl(F_2))$~~has colors~~ 
$(c_1(ll(F_1),\phi_1(c_2(fl(F_2)))=$ 
$(3,\phi_1(1))=(3,5)$ or 
$(3,\phi_2(1))=(3,4)$
 which are proper.
Moreover, for every color $x\ne 2$, $\phi_1(x)\ne\phi_2(x)$. 
Hence, if $c_2(ll(F_2))\ne 2$, then
$\phi_1(c_2(ll(F_2)))\ne \phi_2(c_2(ll(F_2)))$, so 
$d_1(ll(F_2)) \ne d_2(ll(F_2))$.

If $c_2(ll(F_2))=2$, then we make change and set $\phi_2(x)=4x+4$.
The new $\phi_2$ also gives a proper coloring $d_2$, and
$d_1(ll(F_2))\ne d_2(ll(F_2))$, if $c_2(ll(F_2))=2$,.
For all other cases the definitions of $\phi_1$ and $\phi_2$ are given
in Table~\ref{tb1}.
For every $\phi_i$ from the table, the arc $s_1=(r(F_1),r(F_2))$ has 
colors $(0,\phi_i(0))=(0,1)$, $(0,2)$ or $(0,4)$ which are proper.
It is easy to see that in every case, the coloring of the arc $s_2$,
with color $c_1(ll(F_1))$ on one end and $\phi_i(c_2(fl(F_2)))$ on the
other, is proper.
In every line in the table, $\phi_1(x)\ne\phi_2(x)$, 
for every $x\ne0$. Hence, 
$d_1(ll(F_2))=\phi_1(c_2(ll(F_2)))\ne \phi_2(c_2(ll(F_2)))=d_2(ll(F_2))$,
for every $c_2(ll(F_2))$. From the lemma assumpitons, $c_2(ll(F_2))\ne 0$.
\end{proof}

\begin{table*}[thb]
\centering
\begin{tabular}{c|c||c|c|c||c|c}
\hline
& & & & &$s_2$ colored & $s_2$ colored\\
$c_1(ll(F_1))$ & $c_2(fl(F_2))$ & $s_2$ & $\phi_1(x)$ & $\phi_2(x)$ & by $d_1$ & by $d_2$\\
\hline
\hline

1&	1&	$ll(F_1) \rightarrow fl(F_2)$&	
$x+1$&	$x+2$&	$1 \rightarrow 2$&	$1 \rightarrow 3$\\
1&	1&	$ll(F_1) \leftarrow fl(F_2)$&	
$2x+2$&	$2x+4$&	$1 \leftarrow 4$&	$1 \leftarrow 6$\\
1&	3&	$ll(F_1) \rightarrow fl(F_2)$&	
$2x+4$&	$4x+4$&	$1 \rightarrow 3$&	$1 \rightarrow 2$\\
1&	3&	$ll(F_1) \leftarrow fl(F_2)$&	
$x+1$&	$x+4$&	$1 \leftarrow 4$&	$1 \leftarrow 0$\\
\hline
\hline

3& 1& $ll(F_1) \rightarrow fl(F_2)$&\multicolumn{4}{c}{considered separately}\\
3&	1&	$ll(F_1) \leftarrow fl(F_2)$&	
$4x+2$&	$4x+4$&	$3 \leftarrow 6$&	$3 \leftarrow 1$\\
3&	3&	$ll(F_1) \rightarrow fl(F_2)$&	
$x+1$&	$x+2$&	$3  \rightarrow 4$&	$3  \rightarrow 5$\\
3&	3&	$ll(F_1) \leftarrow fl(F_2)$&	
$4x+1$&	$4x+4$&	$3 \leftarrow 6$&	$3 \leftarrow 2$\\
\hline
\end{tabular}
\caption{Definition of $\phi_1(x)$ and $\phi_2(x)$.}
\label{tb1}
\end{table*}


\begin{lem}\label{fan}
For every fan $F$, there is a coloring $c:F\to T_7$ such that
$c(r(F))=0$, $c(fl(F))\ne0$, and $c(ll(F))\ne0$.
\end{lem}
\begin{proof}
Proof by induction on $n$ --- the number of vertices of $F$.
If $n=2$, then $F$ consists of the root  and one leaf and the lemma 
is obvious.
If $n\ge 3$, take $r$ --- the root of $F$ and let $x_1, \dots x_k$ be 
its sons. We have three cases.
\begin{enumerate}
\item $k=1$,
\item $k\ge 2$ and the last son $x_k$ belongs to the exterior path,
\item $k\ge 2$ and the last son $x_k$ belongs to the interior tree.
\end{enumerate}

Case 1. $r$ has only one son $x_1$. This son $x_1$ belongs to the interior of $F$,
because $F$ has more than two vertices. Let $F_1$ be the fan rooted in
 $x_1$ and suppose that arc is going from $r$ to $x_1$.
By induction, there is a coloring $c_1:F_1\to T_7$ such that
$c_1(x_1)=0$, and $c_1(fl(F_1))$, $c_1(ll(F_1))\ne0$. Now define
the coloring $c:F\to T_7$ as follows:

$$c(x) =\left\{
\begin{array}{ll}
 0 &\textrm{ if $x=r$}\\
 c_1(x)+b &\textrm{ if $x\in F_1$}
\end{array}
\right.
$$

where $b\in\{1,2,4\}$ is a constant  satisfying conditions 
$c_1(fl(F_1))+b\ne0$ and
$c_1(ll(F_1))+b\ne0$, such $b$ exists.

Case 2. Let $F_1$ be the fan obtained by removing $x_k$.
By induction, there is a proper  coloring $c_1:F_1\to T_7$.
The vertex $x_k$ is connected by arcs only with $r$ (having color 
$c_1(r)=0$) and $ll(F_1)$ (having color $c_1(ll(F_1)\ne0$). It is easy to 
see that $x_k$ can be colored in a proper way.

Case 3. Let $F_2$ be the fan rooted in $x_k$ and $F_1$ the fan obtained 
from $F$ by removing $x_k$ and its descendants. By induction, $F_1$ and 
$F_2$ can be properly colored and, by Lemma~\ref{comp}, also their 
composition $F$ can be colored in a proper way.
\end{proof}
\section{Main result}
\begin{thm}
Every oriented Halin $H$ graph can be colored with 8 colors.
\end{thm}
\begin{proof}
If $H$ has only 3, 4, or 5 vertices on the exterior cycle, then we 
can color  them with at most five colors, each vertex with different 
color, and the interior tree with additional three colors. 
Hence, in the sequel we shall consider Halin graphs with at least six 
vertices on the exterior cycle.

If not all arcs on the exterior cycle are going in the same direction, 
then we have three vertices $v_1$, $v_2$, $v_3$ on the exterior cycle and 
arcs $(v_1,v_2)$, $(v_3,v_2)$; and let $r$ be the father of 
$v_2$ in the interior tree ($r$ does not have to be the father of $v_1$ or $v_3$),
see Fig.~\ref{fig2}.
\begin{figure}[tbp]
\centering
\includegraphics{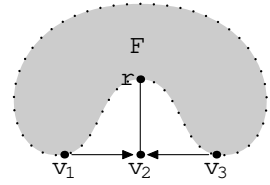}
\caption{Graph with opposite arcs on external cycle.}
\label{fig2}
\end{figure}

Remove $v_2$, and consider the fan $F$ with the root in $r$ and leaves
going from $v_1$ to $v_3$ around the whole graph $H$. 
By Lemma~\ref{fan}, there is a coloring $c:F\to T_7$ such that 
$c(r)=0$, $c(v_1)\ne0$, and $c(v_3)\ne0$. Now we put back $v_2$ with color 7. 

In the sequel we shall consider Halin graphs with all arcs on the exterior 
cycle  going in the same direction. Suppose first that
there are at least two vertices in the interior. Let $r$ be one on 
the lowest level of the interior tree, $p$ be its father in the interior, 
and $x_1, \dots,x_k$ the sons of $r$ on the cycle, see Fig.~\ref{fig4}. Let $x_0$ be the 
predecessor of $x_1$ on the cycle, and $x_{k+1}$ the successor of $x_k$ 
(it is   possible that $x_0=x_{k+1}$).
The arcs on the cycle are going from $x_0$ to $x_{1}$ and so on.
We have four cases:
\begin{enumerate}
\item $k=1$,
\item $k\ge 2$ and there is the arcs $(r,x_1)$,
\item $k\ge 2$ and there is the arcs $(x_k,r)$,
\item $k\ge2$ and there are arcs $(x_i,r)$ and $(r,x_{i+1})$, for some 
$1\le i\le k-1$.
\end{enumerate}

Case 1. If $k=1$, see Fig.~\ref{fig3}, we remove $x_1$ and $r$, and obtain the 
fan with the root in $p$.
\begin{figure}[tbp]
\centering
\includegraphics{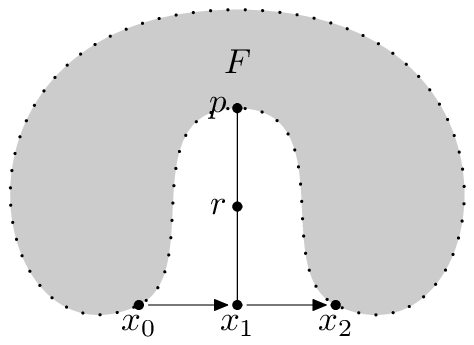}
\caption{Case 1.}
\label{fig3}
\end{figure}
By Lemma~\ref{fan}, we can color it with $c(p)=0$, $c(x_0)\ne0$,
and $c(x_2)\ne0$.
If $c(x_0)\ne c(x_2)$, then we set $x_1$ to 7, and for $r$ we choose 
the color that fits to the color in $p$ and is different from 
$c(x_0)$ and $c(x_2)$ (there are three colors for $r$ that fit to the 
color in $p$).

If $c(x_0)=c(x_2)$,  consider first colors for $r$ and $x_1$ which 
are in accordance only with  two arcs: one  joining $p$ with $r$ and the 
other, joining $r$ with $x_1$. We can obtain at least 6 different colors 
for $x_1$,
and either three of them are proper for the arc $(x_0,x_1)$ or
three are proper for $(x_1,x_2)$. 
In the former case we set $x_2$ to 7, put back $x_1$ and $r$, and color $x_1$ and $r$ in such a way that color in $x_1$ is different from the colors of the neighbors of $x_2$. In the later case we set $x_0$ to 7 and the color in $x_1$ should be different from the colors of the neighbors of $x_0$.

Case 2.
$k\ge 2$ and the arc is going from $r$ to $x_1$, see Fig.~\ref{fig4}.
\begin{figure}[tbp]
\centering
\includegraphics{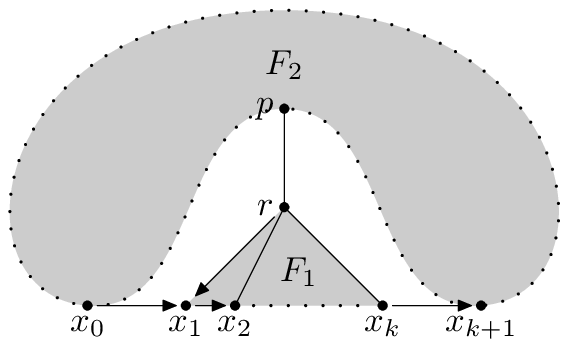}
\caption{Case 2.}
\label{fig4}
\end{figure}
Remove the arc $(x_0, x_1)$. We have two fans:
$F_1$  formed by $r$ and its sons on exterior cycle, and $F_2$ formed 
by $p$ and the other  vertices in $H$.
By  Lemma~\ref{fan}, they can be colored by $T_7$.
Now compose the fans by adding the arc $s_1$ between $p$ and $r$
and $s_2=(x_k,x_{k+1})$.
By  Lemma~\ref{comp}, there are two colorings 
$$d_1,d_2:F_1+F_2\to T_7$$
such that:
\begin{itemize}
\item $d_1(r)=d_2(r)=0$, 
\item $d_1(x_1)=d_2(x_1)\ne0$, 
\item $d_1(x_2)=d_2(x_2)\ne0$,
\item $d_1(x_0)\ne d_2(x_0)$.
\end{itemize}
If $d_1(x_0)=0$ (or $d_2(x_0)=0$), then we simply put the arc 
$(x_0, x_1)$ back.
If $d_1(x_0)\ne0$ and $d_2(x_0)\ne0$, then we choose coloring,
say $d_1$, which gives $d_1(x_0)\ne d_1(x_2)$, put the arc 
$(x_0, x_1)$ back, and set $x_1$ to 7.

Case 3.
$k\ge 2$ and the arc is going from $x_k$ to $r$, see Fig.~\ref{fig5}.
\begin{figure}[tbp]
\centering
\includegraphics{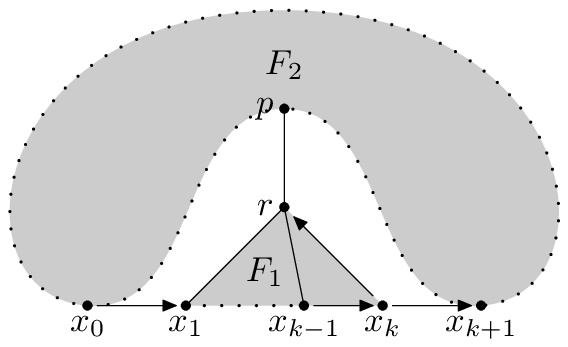}
\caption{Case 3.}
\label{fig5}
\end{figure}
Remove the arc $(x_k, x_{k+1})$. Again consider fans
$F_1$ formed by $r$ and its sons and $F_2$ formed by $p$ and the
other vertices in $H$.
But now they are composed in a different way.
We add the arc $s_1$ between $p$ and $r$
and $s_2=(x_0,x_1)$.
By  Lemma~\ref{comp}, there are two colorings 
$$d_1,d_2:F_1+F_2\to T_7$$
such that: 
\begin{itemize}
\item $d_1(r)=d_2(r)=0$, 
\item $d_1(x_k)=d_2(x_k)\ne0$, 
\item$d_1(x_{k-1})=d_2(x_{k-1})\ne0$,
\item$d_1(x_{k+1})\ne d_2(x_{k+1})$.
\end{itemize}
If $d_1(x_{k+1})=0$ (or $d_2(x_{k+1})=0$), then we simply put the arc 
$(x_k, x_{k+1})$ back.
If $d_1(x_{k+1})\ne0$ and $d_2(x_{k+1})\ne0$, then we choose coloring,
say $d_1$, which gives $d_1(x_{k-1})\ne d_1(x_{k+1})$, put the arc 
$(x_k, x_{k+1})$ back, and set $x_k$ to 7.

Case 4. 
$k\ge2$ and there are arcs $(x_i,r)$ and $(r,x_{i+1})$, for some $1\le i\le k-1$, see Fig.~\ref{fig6}.
\begin{figure}[tbp]
\centering
\includegraphics{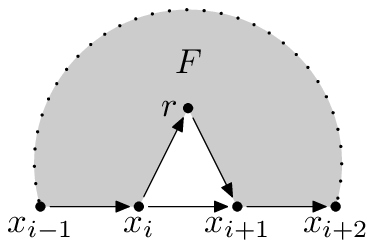}
\caption{Case 4.}
\label{fig6}
\end{figure}
We remove the arc $x_{i}\to x_{i+1}$
and obtain the fan $F$ with the root in $r$ and leaves
going from $x_i$ to $x_{i+1}$ around the whole graph $H$. 
By Lemma~\ref{fan}, there is a coloring $c:F\to T_7$ such that 
$c(r)=0$. Observe that none of $x_{i-1}$, $x_{i}$, $x_{i+1}$, $x_{i+2}$,
 is set to 0.
If $c(x_{i-1})\ne c(x_{i+1})$, then we  can put back the edge 
$(x_{i},x_{i+1})$ and set $x_{i}$ to 7. 
If $c(x_{i})\ne c(x_{i+2})$, then we  can put back the edge 
$(x_{i},x_{i+1})$ and set $x_{i+1}$ to 7. 
If $c(x_{i})=c(x_{i+2})$ and $c(x_{i-1})=c(x_{i+1})$, then $c(x_{i+2})\in \{3,5,6\}$ and we can 
change the color in $x_{i+1}$ (without changing other colors) and set color of $x_i$ to 7.

What is left is the case where there is only one vertex $r$ in the interior
and  all arcs on the exterior cycle are going in the same direction. There are two cases:

\begin{itemize}
\item all arcs incident with $r$ are going in one direction.
Then color exterior with 5 colors and add sixth color to $r$,
\item not all arcs incident with $r$ are going in the same direction.
Then we have the situation described in Case 4 above.
\end{itemize}
\end{proof}


\begin{thebibliography}{99}
\bibitem{biel:i:ann}
H. Bielak, 
The oriented chromatic number of some grids,
{\itshape Annales UMCS Informatica AI} 5
(2006),
5-–17.

\bibitem{bor:ii:dis}
O. V. Borodin, A. V. Kostochka, J. Ne\v{s}e\v{r}il, A. Raspaud, E. Sopena,
On the maximum average degree and the oriented chromatic number of a graph,
{\itshape Discrete Math. 206}
(1999),
77--89.

\bibitem{dol:ii:dis}
M. H. Dolama, E. Sopena,
On the oriented chromatic number of graphs with given excess,
{\itshape Discrete Math. 306}
(2006), 
1342--1350.

\bibitem{dyn:i:ipl}
J. Dybizba\'nski, A. Nenca,
Oriented chromatic number of grids is greater than 7,
{\itshape Inform. Process. Lett. 112}
(2012),
113–-117.

\bibitem{fer:i:ipl}
G. Fertin, A. Raspaud, A. Roychowdhury,
On the oriented chromatic number  of grids,
{\itshape Inform. Process. Lett. 85} 5
(2003),
261--266.

\bibitem{Fr70} 
E. Fried, 
On homogeneous tournaments, 
{\itshape Combinatorial theory and its applications}, Vol. II (ed. P.~Erd\"os et al.), North-Holland,
Amsterdam (1970), 467--476.


\bibitem{hos:ii:ipl}
M. Hosseini Dolama, E. Sopena,
On the oriented chromatic number of Halin graphs,
{\itshape Inform. Process. Lett. 98} 6
(2006)
6,
247--252.  

\bibitem{kos:ii:gra}
A. V. Kostochka, E. Sopena, X. Zhu,
Acyclic and oriented chromatic numbers of graphs,
{\itshape J.Graph Theory 24(4)}
(1997),
331--340.

\bibitem{ras:ii:ipl}
A. Raspaud, E. Sopena,
Good and semi-strong colorings of oriented planar graphs,
{\itshape Inform. Process. Lett. 51} 4 
(1994)
171--174.


\bibitem{so2:ii:dis}
E. Sopena,
Oriented graph coloring,
{\itshape Discrete Math. 229}
(2001),
359--369.

\bibitem{Sop:i:pre} 
E. Sopena, 
The oriented chromatic number of graphs: A short survey, 
preprint 
(2013).

\bibitem{so1:ii:gra}
E. Sopena,
The chromatic number of oriented graphs,
{\itshape J.Graph Theory 25}
(1997),
191--205.

\bibitem{so4:ii:ipl}
E. Sopena,
There exist oriented planar graphs with oriented chromatic number at least sixteen,
{\itshape  Inform. Process. Lett. 81} 6
(2002),
309--312.

\bibitem{so3:ii:bor}
E. Sopena, L. Vignal,
A note on the oriented chromatic number of graphs with maximum degree three,
{\itshape Research Report, Bordeaux I University},
(1996).

\bibitem{sze:ii:ipl}
A. Szepietowski, M. Targan,
A note on the oriented chromatic number of grids,
{\itshape  Inform. Process. Lett. 92} 2
(2004),
65--70.

\bibitem{Vig}
L. Vignal, 
Homomorphismes et coloration de graphes orientés (in French),
{\itshape PhD Thesis}, 
University of Bordeaux 1, 
(1997) .






\end{thebibliography}
\end{document}